\documentclass{amsart}
\usepackage{amsmath,amsfonts}
\usepackage{algorithmic}
\usepackage{algorithm}
\usepackage{array}
\usepackage[caption=false,font=normalsize,labelfont=sf,textfont=sf]{subfig}
\usepackage{textcomp}
\usepackage{stfloats}
\usepackage{url}
\usepackage{verbatim}
\usepackage{graphicx}
\usepackage{cite}

\newif\ifmaintext
\maintextfalse

\newcommand{\preprinttitle}{Equivalence of Informations Characterizes Bregman Divergences}

\newcommand{\parstart}[2]{
    \ifmaintext 
        \IEEEPARstart{#1}{#2}%
    \else
        #1#2%
    \fi}

\usepackage{amsthm}
\usepackage{amsmath}
\usepackage{amssymb}
\usepackage{cleveref}

\usepackage{booktabs}
\usepackage{graphicx}

\newcommand{\R}{\mathbb{R}}

\newcommand{\mW}{\mathbf{W}}
\newcommand{\mI}{\mathbf{I}}

\newcommand{\mX}{\mathbf{X}}
\newcommand{\abs}[1]{\left| #1 \right|}

\newcommand{\E}{\mathbb{E}}

\newcommand{\cC}{\mathcal{C}}
\newcommand{\cH}{\mathcal{H}}
\newcommand{\cA}{\mathcal{A}}
\newcommand{\cB}{\mathcal{B}}

\newcommand{\vdelta}{\boldsymbol{\delta}}

\newcommand{\vx}{\mathbf{x}}
\newcommand{\vy}{\mathbf{y}}
\newcommand{\vv}{\mathbf{v}}
\newcommand{\vw}{\mathbf{w}}
\newcommand{\vp}{\mathbf{p}}
\newcommand{\vq}{\mathbf{q}}
\newcommand{\norm}[1]{\left\lVert #1 \right\rVert}
\newcommand{\simplex}[1]{\Delta_{#1}}

\newcommand{\paren}[1]{\left( #1 \right)}
\newcommand{\braces}[1]{\left\{ #1 \right\}}
\newcommand{\brackets}[1]{\left[ #1 \right]}

\newcommand{\sumi}{\sum_{i=1}^n}
\newcommand{\sumj}{\sum_{j=1}^\ell}

\theoremstyle{definition}
\newtheorem{definition}{Definition}

\theoremstyle{theorem}
\newtheorem{theorem}{Theorem}
\newtheorem{lemma}{Lemma}

\theoremstyle{plain}

\begin{document}

\title{\preprinttitle}
\date{\today}

\author{Philip S. Chodrow}
\thanks{Department of Computer Science, Middlebury College, Middlebury, Vermont, USA}

\begin{abstract}
    Bregman divergences are a class of distance-like comparison functions which play fundamental roles in optimization, statistics, and information theory. 
One important property of Bregman divergences is that they cause two useful formulations of information content (in the sense of variability or non-uniformity) in a weighted collection of vectors to agree. 
In this note, we show that this agreement in fact characterizes the class of Bregman divergences; they are the only divergences which generate this agreement for arbitrary collections of weighted vectors. 

\end{abstract}

\maketitle

\section{Introduction}

\parstart{F}{or} a convex set $\cC \subseteq \R^\ell$ with relative interior $\cC_*$ and a strictly convex function $\phi: \cC \to \R$ differentiable on $\cC_*$, the Bregman divergence induced by $\phi$ is the function $d_\phi: \cC \times \cC_* \to \R$ defined by
\begin{align}
    d_\phi(\vx_1,\vx_2) = \phi(\vx_1) - \phi(\vx_2) - \nabla \phi(\vx_2)^T(\vx_1-\vx_2)\;. \label{eq:bregman-divergence}
\end{align}
Two common examples of Bregman divergences are: 
\begin{itemize}
    \item The squared Mahalanobis distance $d_\phi(\vx_1, \vx_2) = (\vx_1 - \vx_2)^T \mW (\vx_1 - \vx_2)$, where $\mW$ is a positive-definite matrix. 
    The function $\phi$ corresponding to this case is given by $\phi(\vx) = \frac{1}{2}\vx^T\mW \vx$. 
    The special case $\mW = \mI$ gives the squared Euclidean distance.
    This divergence may be defined on $\cC = \R^\ell$.
    \item The Kullback-Leibler (KL) divergence $d_\phi(\vx, \vy) = \sumi x_{i} \log \frac{x_{i}}{y_i}$, where $\vp$ and $\vq$ are probability vectors. 
    The KL divergence is induced by the negative entropy function $\phi(\vp) = \sumi x_i \log x_i$.  
    The KL divergence is defined on the probability simplex $\simplex{\ell} = \braces{\vp \in \R^\ell \;|\; \sumi x_i = 1\;,\; x_i \geq 0 \; \forall i}$.
    Extensions are possible to general convex subsets of $\R_+^\ell$.
    When computing the KL divergence, we use the convention $0\log 0 = 0$. 
\end{itemize}

Bregman divergences \cite{bregmanRelaxationMethodFinding1967} play fundamental roles in information theory, statistics, and machine learning; see \cite{reemReexaminationBregmanFunctions2019} for a review.  
Like metrics, Bregman divergences are positive-definite: $d_\phi(\vx,\vy) \geq 0$ with equality if and only if $\vx = \vy$.
Unlike metrics, Bregman divergences are not in general symmetric and do not in general satisfy a triangle inequality. 
Bregman divergences are locally distance-like in that they induce a Riemannian metric on $\cC$ obtained by the small-$\vdelta$ expansion
\begin{align}
    d_\phi(\vx + \vdelta, \vx) = \frac{1}{2}\vdelta^T\cH\phi(\vx)\vdelta + o(\norm{\vdelta}^2)\;,
\end{align}
where $\vdelta$ is a small perturbation vector and $\cH\phi(\vx)$ is the Hessian of $\phi$ at $\vx$.
Because $\phi$ is strictly convex, $\cH\phi(\vx)$ is positive-definite and defines a Riemannian metric on $\cC$ \cite{amariInformationGeometryDivergence2010}. 

Bregman divergences provide one natural route through which to generalize Shannon information theory, with the function $-\phi$ taking on the role of the Shannon entropy. 
Although there are multiple characterization theorems for fundamental information-theoretic quantities such as entropy \cite{baezCharacterizationEntropyTerms2011,faddeev1956concept,shannonMathematicalTheoryCommunication1948}, mutual information \cite{fullwoodAxiomaticCharacterizationMutual2023,frankelMeasuringSchoolSegregation2011}, and the Kullback-Leibler divergence \cite{jiaoInformationMeasuresCurious2014,hobsonNewTheoremInformation1969}, the author is aware of only one extant characterization of the more general class of Bregman divergences \cite{banerjeeOptimalityConditionalExpectation2005}: Bregman divergences are the unique class of loss functions which render conditional expectations loss-minimizing in stochastic prediction problems. 

In this short note, we prove a new characterization of the class of Bregman divergences, based on an equality of two common formulations of information content in weighted collections of vectors. 

\section{Bregman Divergence and Two Informations}

In this section, we state two standard formulations of the concept of information contained in a data set and discuss the role of Bregman divergences in relating them. 
Let $\simplex{n}$ be the set of discrete probability measures on $n$ points. 
We treat each formulation of information as a function $\simplex{n} \times \cC^{n} \to \R$. 
The first standard formulation of information content consideres a weighted sum of strictly convex loss functions, which is compared to the same loss function evaluated at the data centroid.

\begin{definition}[Jensen Gap Information]
    Let $\phi: \cC \to \R$ be a strictly convex function on $\cC$.
    The \emph{Jensen gap information} is the function $I_\phi: \simplex{n} \times \cC^{\ell} \to \R$ with given by 
    \begin{align}
        I_\phi(\mu, \mX) \triangleq \sumi \mu_i \phi(\vx_i) - \phi \paren{\vy}\;,
    \end{align}
    where $\vy = \sumi \mu_i \vx_i$.  
\end{definition}

If we define $X$ to be a random vector that takes value $\vx_i$ with probability $\mu_i$, Jensen's inequality states that $\E[\phi(X)] \geq \phi(\E[X])$, with equality holding only if $X$ is constant (i.e. if there exists $i$ such that $\mu_i = 1$). 
The Jensen gap information is a measure of the difference of the two sides of this inequality; indeed, $\E[\phi(X)] = \phi(\E[X]) + I_\phi(\mu, X)$ \cite{banerjeeClusteringBregmanDivergences2004,banerjeeOptimalityConditionalExpectation2005}.
This formulation makes clear that $I_\phi$ is nonnegative and that $I_\phi(\mu, \mX) = 0$ if and only if $\mX$ is constant on the rows supported by $\mu$.

Another standard concept of information content involves a weighted mean of divergences from the centroid. 

\begin{definition}[Divergence]
    A function $d: \cC \times \cC \to \R$ is a \emph{divergence} if $d(\vx_1,\vx_2) \geq 0$ for any $\vx_1, \vx_2 \in \cC$,  with equality if and only if $\vx_1 = \vx_2$. 
\end{definition}

\begin{definition}[Divergence Information]
    Let $d$ be a divergence. 
    The \emph{divergence information} is the function $I_d: \simplex{n} \times \cC^{n} \to \R$ given by
    \begin{align}
        I_d(\mu, \mX) \triangleq \sumi \mu_i d(\vx_i, \vy)\;, \label{eq:divergence-information}
    \end{align}
    where $\vy = \sumi \mu_i \vx_i$. 
\end{definition}
In this definition, we assume that $\vy \in \cC_*$; as noted by \cite{banerjeeClusteringBregmanDivergences2004}, this assumption is not restrictive since the set $\cC$ can be replaced with the convex hull of the data $\mX$ without loss of generality.  
The divergence information measures the $\mu$-weighted average divergence of $\vx_i$ from the centroid $\vy$. 
The divergence information is related to the only extant characterization result for Bregman divergences known to the author: a divergence $d$ is a Bregman divergence if and only if the vector $\vy = \sumi {\mu_i}\vx_i$ is the unique minimizer of the righthand side of \cref{eq:divergence-information} for any choice of $\mu$ and $\mX$ \cite{banerjeeOptimalityConditionalExpectation2005}. 

There are several important cases in which the Jensen gap information and the divergence information coincide. 

\begin{definition}[Information Equivalence]
    We say that a pair $(\phi, d)$ of a strictly convex function $\phi: \cC \rightarrow \R$ and a divergence $d: \cC \times \cC \to \R$ satisfies the \emph{information equivalence property} if, for all $(\mu, \mX) \in \simplex{n} \times \cC^n$, it holds that 
    \begin{align}
        I_\phi(\mu, \mX) = I_d(\mu, \mX)\;. \label{eq:information-equivalence}
    \end{align}
\end{definition}

\begin{lemma}[Information Equivalence with Bregman Divergences\cite{banerjeeOptimalBregmanPrediction2004,banerjeeClusteringBregmanDivergences2004}] \label{lm:bregman-agreement}
    If $d = d_\phi$, then the pair $(\phi, d)$ satisfies the information equivalence property.
\end{lemma}
The proof of this lemma is a direct calculation and is provided in \cite{banerjeeClusteringBregmanDivergences2004}. 
When $\phi(\vx) = \frac{1}{2}\norm{\vx}^2$ and $d = d_\phi$ is the Euclidean distance, the information equivalence property \cref{eq:information-equivalence} is equivalent to identity 
\begin{align}
    \sumi \mu_i \norm{\vx_i}^2 - \norm{\sumi \mu_i \vx_i}^2 = \sumi \mu_i \norm{\vx_i - \sumi \mu_i \vx_i}^2\;. \label{eq:euclidean-agreement}
\end{align}
The righthand side of \cref{eq:euclidean-agreement} is the weighted sum-of-squares loss of the data points $\vx_i$ with respect to their centroid $\sumi \mu_i \vx_i$, which is often used in statistical tests and clustering algorithms. 
\Cref{eq:euclidean-agreement} asserts that this loss may also be computed from a weighted average of the norms of the data points. 

When $\cC$ is the probability simplex, $\phi(\vx) = \sumi x_i \log x_i$ is the negative entropy, and $d = d_\phi$ is the KL divergence, the information equivalence property \cref{eq:information-equivalence} expresses the equality of two equivalent formulations of the mutual information for discrete random variables. 
Let $A$ and $B$ be discrete random variables on alphabets $\cA$ of size $k$ and $\cB$ of size $\ell$ respectively. 
Suppose that their joint distribution is $p_{A,B}(a_i,b_j) = \mu_i x_{ij}$. 
Let $\vy$ be the vector with entries $y_j = \sumi \mu_i x_{ij}$; then $\vy$ is the marginal distribution of $B$.
The Jensen Gap information $I_\phi(\mu, \mX)$ is 
\begin{align}
    I_\phi(\mu, \mX) &= \underbrace{\sumi \mu_i \sumj x_{ij} \log x_{ij}}_{-H(B|A)} - \underbrace{\sumj y_j \log y_j}_{-H(B)};
\end{align}
which expresses the mutual information $I(A;B)$ between random variables $A$ and $B$ in the entropy-reduction formulation, $I(A;B) = H(B) - H(B|A)$ \cite{cover2012elements}.
On the other hand, the divergence information $I_d(\mu,\mX)$ is 
\begin{align}
    I_d(\mu, \mX) = \sumi\mu_i \underbrace{\sumj x_{ij} \log \frac{x_{ij}}{y_j}\;,}_{d_\phi(\vx_i, \vy)}
\end{align}
which expresses the mutual information $I(A;B)$ instead as the weighted sum of KL divergences of $\vx_{i}$ from $\vy$.

Our contribution in this paper is to prove a converse to \Cref{lm:bregman-agreement}: the Bregman divergence $d_\phi$ is the \emph{only} divergence which satisfies information equivalence with $\phi$.

\section{Main Result}

\begin{theorem} \label{thm:main}
    If $d$ is a divergence and if the pair $(\phi, d)$ satisfies the information equivalence property \cref{eq:information-equivalence}, then $d$ is the Bregman divergence induced by $\phi$: $d(\vx,\vy) = d_\phi(\vx,\vy)$ for any $\vx \in \cC$ and $\vy \in \cC_*$.
\end{theorem}

    For any $\vx \in \cC$ and $\vy \in \mathrm{int}\;\cC$, we can write
    \begin{align}
        d(\vx,\vy) = \phi(\vx) - \phi(\vy) + f(\vx,\vy) \label{eq:introducing-f}
    \end{align}
    for some unknown function $f:\cC \times \cC_*\rightarrow \R$.

    Our first step is to show that the condition \cref{eq:information-equivalence} implies that $f$ is an affine function of its first argument $\vx$ on $\cC$.
    To do so, we observe that if $\mu \in \simplex{n}$ and $\mX \in \cC^n$ are such that $\sumi \mu_i \vx_i = \vy$, then information equivalence \eqref{eq:information-equivalence} enforces that $I_\phi(\mu,\mX) = I_d(\mu,\mX)$ for all $\mu \in \simplex{n}$. 
    This means that we must have 
    \begin{align*}
        \sumi\mu_i \phi(\vx_i) - \phi(\vy) &= \sumi\mu_i d(\vx_i,\vy) \\ 
                                          &= \sumi\mu_i\brackets{\phi(\vx_i) - \phi(\vy) + f(\vx_i,\vy)} \\ 
                                          &= \sumi\mu_i\phi(\vx_i) - \phi(\vy) + \sumi\mu_i f(\vx_i,\vy) \;,
    \end{align*}
    from which it follows that 
    \begin{align}
        \sumi\mu_i f(\vx_i,\vy) = 0 \label{eq:mean-f-is-0}\;.
    \end{align}

    Consider the function $g_\vy(\vv) = f(\vv + \vy, \vy)$. 
    The condition \cref{eq:mean-f-is-0} implies that 
    \begin{align}
        \sumi \mu_i g_\vy(\vv_i) = 0 \label{eq:mean-g-is-0}\;.
    \end{align}
    for any $\vv_1,\ldots,\vv_k$ such that $\sumi \mu_i \vv_i = 0$.

    To show that $f$ is affine, it suffices to show that the function $g_\vy$ is linear. 
    We do this through a sequence of short lemmas. 

    \begin{lemma}\label{lm:tech-3}
        The function $g_\vy$ satisfies $g_\vy(-\vv) = -g_\vy(\vv)$. 
    \end{lemma}
    \begin{proof}
        By \cref{eq:mean-g-is-0}, we have that 
        \begin{align*}
            \frac{1}{2}g_\vy(\vv) + \frac{1}{2}g_\vy(-\vv) = 0\;,
        \end{align*}
        from which the lemma follows. 
    \end{proof}

    \begin{lemma} \label{lm:almost-linear}
        Let $\tilde{\gamma}_1,\ldots \tilde{\gamma}_k$ be a set of nonnegative scalars such that $\sumi \tilde{\gamma}_i = 1$. 
        Then, for any vectors $\vv_1,\ldots,\vv_k$, it holds that 
        \begin{align}
            g_\vy\paren{\sumi \tilde{\gamma}_i \vv_i} = \sumi \tilde{\gamma}_i g_\vy(\vv_i)\;.
        \end{align}
    \end{lemma}
    \begin{proof}
    
        Let $c_i = -\frac{1}{2}\tilde{\gamma}_i$. 
        Let $\vw = -\sumi\tilde{\gamma}_i \vv_i$. 
        By construction, we have $\frac{1}{2} - \sumi c_i = 1$ and $\frac{1}{2}\vw  - \sumi c_i \vv_i  = 0$. 
        Noting that $c_i < 0$ and using \cref{eq:mean-g-is-0} , we infer 
        \begin{align}
            \frac{1}{2} g_\vy(\vw) - \sumi c_i g_\vy(\vv_i) = 0\;. \label{eq:intermediate-1}
        \end{align}
        We then have 
        \begin{align}
            g_\vy\paren{-\sumi \tilde{\gamma}_i \vv_i} &= -g_\vy\paren{\sumi \tilde{\gamma}_i \vv_i} \tag{\Cref{lm:tech-3}}\\
            &= -2\sumi c_i g_\vy(\vv_i) \tag{by \cref{eq:intermediate-1}}\\
            &= \sumi \tilde{\gamma}_i g_\vy(\vv_i)\;, \nonumber
        \end{align}
        as was to be shown. 
    \end{proof}
    
    \begin{lemma}\label{lm:tech-4}
        For any vector $v$ and scalar $c$, we have $g_\vy(cv) = cg_\vy(v)$.
    \end{lemma}
    \begin{proof}
        We proceed by cases.
        \begin{enumerate}
            \item \textbf{$c = 0$}. \Cref{lm:tech-3} implies that $g_\vy(0) = 0$.
            \item \textbf{$c > 0$}. Define $\mu_1 = \frac{c}{1+c}$ and $\mu_2 = \frac{1}{1 + c}$. 
            Let $\vw =  - c\vv$. 
            Then, by construction, we have $\mu_1 + \mu_2 = 1$ and $\mu_1\vv + \mu_2\vw = 0$.  
            It follows that $\mu_1 g(\vv) + \mu_2 g(\vw) = 0$,
            Isolating $g(\vw)$, we have 
            \begin{align*}
                g(\vw) = -\frac{\mu_1}{\mu_2}g(\vv) = -cg(\vv)\;.
            \end{align*}
            Recalling that $\vw = -c\vv$ and applying \Cref{lm:tech-3} completes the proof of this case. 
            \item \textbf{$c < 0$}. This case follows by applying the proof of the previous case, replacing $c$ with $-c$.
        \end{enumerate} 
    \end{proof}

    \begin{lemma} \label{lm:tech-2}
        The function $g$ is linear.
    \end{lemma}
    
    \begin{proof}

        Fix arbitrary coefficients $\gamma_1,\ldots,\gamma_k$ and vectors $\vv_1,\ldots,\vv_k$. 
        Let 
        \begin{align*}
            \tilde{\gamma}_i =  \frac{\abs{\gamma_i}}{\sumi \abs{\gamma_i}} \quad \text{,} \quad s_i =\begin{cases}
                \frac{\gamma_i}{\tilde{\gamma}_i} &\quad \gamma_i \neq 0\\
                0 &\quad \gamma_i = 0\;,
            \end{cases}
            \quad \text{and} \quad 
            \tilde{\vv}_i = s_i \vv_i\;.
        \end{align*}
        Then, by construction, it holds that $\tilde{\gamma}_i \geq 0$ for all $i$,  $\sumi \tilde{\gamma}_i = 1$, and 
        \begin{align*}
            \sumi \tilde{\gamma}_i \tilde{\vv}_i = \sumi \gamma_i \vv_i\;.
        \end{align*}
        The conditions of \Cref{lm:almost-linear} are satisfied, and we have
        \begin{align}
            g_\vy\paren{\sumi \gamma_i \vv_i} &= g_\vy\paren{\sumi \tilde{\gamma}_i \tilde{\vv}_i} \nonumber\\
                                           &= \sumi \tilde{\gamma}_i g_\vy(\tilde{\vv}_i) \nonumber \tag{\Cref{lm:almost-linear}}\\ 
                                           &= \sumi \tilde{\gamma}_i s_i g_\vy(\vv_i) \tag{\Cref{lm:tech-4}} \nonumber \\ 
                                           &= \sumi \gamma_i g_\vy(\vv_i) \nonumber \;. 
        \end{align}
        This completes the proof. 
    \end{proof}

\begin{proof}[Proof of \Cref{thm:main}]
    The preceding lemmas prove that $g_\vy$ is linear. 
    Since for constant $\vy$ the function $f$ in \cref{eq:introducing-f} is a translation of $g_\vy$ in its first argument, it follows that $f$ is affine as a function of its first argument. 
    We may therefore write 
    \begin{align}
        f(\vx,\vy) = h_1(\vy)^T\vx + h_2(\vy)\;.  \label{eq:linear-f}
    \end{align}
    for some functions $h_1:\cC_* \rightarrow \R^\ell$ and $h_2:\cC_*\rightarrow \R$. 
    
    We now determine these functions. 
    First, since $\phi$ is differentiable on $\cC*$ and $f(\vx,\vy)$ is affine in $\vx$, $d(\vx,\vy)$ is differentiable in its first argument on $\cC_*$. 
    Since $d$ is a divergence, it is positive-definite and therefore $\vy$ is a critical point of the function $d(\cdot,\vy)$ on $\cC_*$.
    It follows that $\nabla_1 d(\vy,\vy)$, the gradient of $d$ with respect to its first argument, is orthogonal to $\cC_*$ at $\vy$: 
    \begin{align}
        \nabla_1 d(\vy, \vy)^T(\vx - \vy) = 0 \label{eq:grad-orthogonal}
    \end{align}
    for any $\vx \in \cC$. 
    We can compute $\nabla_1 d(\vy,\vy)$ explicitly; it is $\nabla_1 d(\vy,\vy) = \nabla \phi(\vy) + h_1(\vy)$. 
    \Cref{eq:grad-orthogonal} becomes 
    \begin{align}
        (\nabla \phi(\vy) + h_1(\vy))^T (\vx - \vy) = 0\;. \label{eq:linear-grad-phi}
    \end{align}
    for any $\vx$ and $\vy$. 

    Now, the condition that $d(\vy,\vy) = 0$ implies that $h_2(\vy) = -h_1(\vy)^T\vy$.
    We then compute 
    \begin{align}
        -\nabla \phi(\vy)^T(\vx - \vy) &= h_1(\vy)^T(\vx - \vy)  \tag{\cref{eq:linear-grad-phi}} \\ 
                                 &= h_1(\vy)^T\vx + h_2(\vy)  \nonumber \\ 
                                 &= f(\vx,\vy)\;. \tag{\cref{eq:linear-f}}
    \end{align}
    Recalling the definition of $f$ in \cref{eq:introducing-f}, we conclude that 
    \begin{align*}
        d(\vx,\vy) = \phi(\vx) - \phi(\vy) - \nabla \phi(\vy)^T(\vx - \vy)\;,
    \end{align*}
    which is the Bregman divergence induced by $\phi$.
    This completes the proof. 
\end{proof}

\section{Discussion}

    We have shown that the class of Bregman divergences is the unique class of divergences which induce agreement between the Jensen gap and divergence informations.
    This result offers some further perspective on the role for Bregman divergences in data clustering and quantization \cite{banerjeeClusteringBregmanDivergences2004}. 
    The Jensen gap information $I_\phi$ is a natural loss function for such tasks, with one motivation as follows. 
    Suppose that we wish to measure the complexity of a set of data points $\mX$ with weights $\mu \in \simplex{n}$ using a weighted per-observation loss function and a term which depends only on the centroid $\vy = \sumi \mu_i \vx_i$ of the data: 
    \begin{align*}
        L(\mu, \mX) = \sumi \mu_i \psi(\vx_i) + \rho(\vy)\;. 
    \end{align*}
    A natural stipulation for the loss function $L$ is that replacing two data points $\vx_1$ and $\vx_2$ with their weighted mean $\vx = \frac{\mu_1}{\mu_1 + \mu_2} \vx_1 + \frac{\mu_2}{\mu_1 + \mu_2}\vx_2$ should strictly decrease the loss when $\vx_1 \neq \vx_2$; this requirement is equivalent to strict convexity of the function $\psi$. 
    If we further require that $L(\mu, \mX) = 0$ when each row of $\mX$ is identical, we find that $\rho(\vy) = -\psi(\vy)$ and that our loss function is the Jensen gap information: $L(\mu,\mX) = I_\psi(\mu,\mX)$. 
    The result of this paper shows that this natural formulation fully determines the choice of how to perform pairwise comparisons between individual data points; only the corresponding Bregman divergence can serve as a comparator which is consistent with the Jensen gap information.

\bibliographystyle{plain}
\bibliography{bregman.bib}

\vspace{11pt}

\end{document}